\documentclass{article}
\usepackage{authblk}
\usepackage[utf8]{inputenc}
\usepackage{graphicx}
\usepackage{amsfonts}
\usepackage{amsmath}
\usepackage{nccmath}
\usepackage{amssymb}
\usepackage{fancyhdr}
\usepackage{amsthm}
\usepackage{array}
\usepackage{systeme}
\usepackage{xcolor}
\usepackage{mathtools}

\usepackage{afterpage}
\usepackage{calligra}
\usepackage[T1]{fontenc}
\usepackage{esint}
\usepackage{amsmath}

\usepackage{enumerate}
\usepackage{geometry}
\usepackage{hyperref}
\usepackage{subfig}

\usepackage{comment}
\usepackage{tikz}
\usepackage{scalerel}

\title{On Time-subordinated Brownian Motion Processes for Financial Markets}
\author[1]{Rohan Shenoy}
\author[2]{Peter Kempthorne}
\affil[1]{Department of Mathematics, Imperial College London, rohan.shenoy22@imperial.ac.uk}
\affil[2]{Department of Mathematics, Massachusetts Institute of Technology, kempthorne@math.mit.edu}
\date{October 2025}

\begin{document}

\maketitle

\begin{abstract}

\noindent In the context of time-subordinated Brownian motion models,  Fourier theory and methodology are proposed to modelling the stochastic distribution of time increments. Gaussian Variance-Mean mixtures and time-subordinated models are reviewed with a key example being the Variance-Gamma process. A non-parametric characteristic function decomposition of subordinated Brownian motion is presented. The theory requires an extension of the real domain of certain characteristic functions to the complex plane, the validity of which is proven here. This allows one to characterise and study the stochastic time-change directly from the full process. An empirical decomposition of S\&P log-returns is provided to illustrate the methodology.\\

\end{abstract}

\tableofcontents

\newpage

\newtheorem{thm}{Theorem}
\newtheorem{cor}{Corollary}
\newtheorem{prop}{Proposition}
\newtheorem{lemma}{Lemma}
\newtheorem{definition}{Definition}
\newtheorem{remark}{Remark}
\newtheorem{example}{Example}
\newtheorem*{examples}{Examples}

\section{Introduction}\label{Section1}

Under many simple market models, the price returns of a financial asset are assumed to follow a log-Normal distribution.  The celebrated Black-Scholes-Merton option pricing model \cite{blackscholes}, 
assumes that a stock's log price follows a random walk in continuous time with a variance proportional to the square of the stock price, the Geometric Brownian Motion (GBM) model.
For any finite time interval, the distribution of price increments (i.e. "steps" in the random walk) is log-Normal and the distribution of the log return of stock prices over the time interval is Normal with variance a constant proportion of the length of the interval.

In empirical modelling of asset price dynamics, time series of asset prices are readily available. Commonly, the time frequency of such series is daily, but can be lower with weekly or monthly prices,  or higher with  time bars  which have \textit{fixed-length} time intervals (e.g. 30, 5, 1 minutes).  With higher frequencies, the market micro-structure at the transaction level plays a significant role with possible prices determined by bid and offer prices of market makers that are constrained to discrete price levels, and variation in trading activity over a trading day in terms of trade counts and trade volume.

In this paper, we focus on models of time series of closing prices of an asset with daily or lower frequency. Such time series are well-suited to applying the Geometric Brownian Motion model and extensions without the need to accommodate micro-structure issues of higher frequency data. Indeed, generalized  Central Limit Theorems might apply to model daily log-price increments as Normal since they are the sums of large-sample log-price increments at higher frequencies.

It is common to index daily price observations by market-day counts and assume  unit time increments over all successive price observations. Assuming the GBM model for such time series,  the time series of daily prices is transformed to a daily time series of log returns which is equivalent to a random sample (independent and identically distributed) from a Normal distribution  with mean equal to the daily drift rate and variance equal to the daily variance rate. However goodness-of-fit tests and diagnostics (e.g., Shapiro Wilk Normality Test \cite{ShapiroWilk1965} and Normal QQ Plots) of Normal model fits to  daily log return time series often \textit{reject} the Normal distribution assumption due to "heavy-tailed" empirical distributions that are not consistent with the Normal distribution shape.

To maintain use of the GBM model, many authors have proposed possible alternative time indices to market-day counts, such as shares traded or transaction counts (e.g. Clark(1973) \cite{Clark1973}). Making the time increments between price observations conform with a progress indicator is known as a time-subordination: a change of the original time-scale (cumulative count of market days) to an alternate basis. This new time-scale is often labelled as "business time" \cite{VeraartWinkel}.
 
Significantly,  the works of Madan and Milne (1991) \cite{MadanMilne} and Madan, Carr and Chang (1998) \cite{MadanCarrChang} have highlighted the possibility of improving upon the use of observable  indices by implementing a stochastic index through a random time-subordination. The authors demonstrate that the Variance Gamma model, a random time-subordinated model, provides more accurate pricing performance of call options than Geometric Brownian motion.  Specifically it values options higher, particularly for out-of-the-money options with  long maturity on stocks with high kurtosis.

Under a random time-subordination model, one infers that the economically relevant time in a market is itself a stochastic process, independent of other random processes affecting market dynamics.  This model allows the distribution of price movements itself on different market days to differ randomly because some days are "longer" and other days are "shorter"\footnote{In their paper of 1990, Madan and Seneta comment on this idea, stating: \textit{More informally, one may think of $G(t)$ [the stochastic index] as a formal statement of the remark, "Didn't have much of a year this year," by allowing for an interpretation of how much of a year one actually had.}}. The results of Skorokhod \cite{skororussia} and Monroe \cite{Monroe} (as highlighted in Veraart and Winkel \cite{VeraartWinkel}) demonstrate that in fact any arbitrage free model can be represented as a Brownian motion with random time change, providing more concrete motivation for studying such models as they represent a valuable change of viewpoint.

We study the mathematical model of log returns where the time increments follow a subordinator process  and the log price increments are attributed to Brownian motion on the stochastic time scale. Due to a theorem of Dubins and Schwarz \cite{DubinsSchwarz}, one can express any continuous local martingale $M$ as a time-changed Brownian motion where the time-change is given by the continuous quadratic variation of $M$. However, we find that the full process has a closer relationship with the subordinator process in Fourier space which could provide additional insight. To this end we develop an alternate, Fourier method which enables us to directly transform a time-changed Brownian motion process into its subordinator process. This is formulated as the time-change transform and allows one to reduce the study of the full process to studying just the stochastic time-change in isolation. We define the Fourier transform of the subordinator process from the Fourier transform of data from the full process and apply an inversion back into real space we obtain the probability distribution of the subordinator process.
Empirically, one can then characterise the underlying random subordinator process straight from price observations over market-day counts.

Empirical modelling of the subordinator process directly enables evaluation of its characteristics. A key running example in this paper is the Variance Gamma model of Madan et al. \cite{MadanMilne}, \cite{MadanCarrChang} which assumes the Gamma process for the time subordinator of log-transformed daily price time series.  For this case, the time subordinator process  has stationary increments which are independent over non-overlapping time intervals, properties satisfied by Lévy processes. A significant contribution of this paper is the theory and methods for estimation and hypothesis testing of Gaussian variance-mean mixture distributions for which the Variance-Gamma process is a special case.

Section \ref{SectionGVMM} characterises Gaussian variance-mean mixtures for modelling daily log returns - the special case of these models when the variance distribution is Gamma corresponds to the Variance-Gamma process. We then consider the extension into continuous time-subordinated Brownian motion processes in Section \ref{SectionTSBM}. This allows the formulation of the time-change transform in Section \ref{SectionTSD}.  Section \ref{SectionEA} presents a brief example of applying the time-change transform to the S$\&$P500 index with the, including evaluating its consistency with the Variance-Gamma process for fixed and varied holding periods of log-returns (varying the number of days in holding periods 1,5,10,15,20 days). Section \ref{SectionSummary} concludes the results.

\section{Gaussian Variance-Mean Mixtures}\label{SectionGVMM}
We begin with defining the Gaussian variance-mean mixture distribution, which arises as the distribution of the increment of a Brownian motion process with drift $\theta$ and volatility $\sigma$  over a \textit{stochastic} time increment.

Let $\{Y_t ,  t\geq 0\}$ denote the Brownian motion process with drift $\theta$ and volatility $\sigma$.  For initial time $t$ and a time increment of deterministic length $v\geq0$, the increment  of the Brownian motion:
$ X = Y_{t+v} - Y_t$
is a $ N(\mu_X,\sigma_X^2)$ random variable with mean $\mu_X=\theta v$ and variance $\sigma_X^2 = \sigma^2 v.$  It will be convenient to represent this distribution as an affine transformation of a standard Gaussian random variable, $Z \sim N(0,1)$:
$ X = \theta v + \sigma \sqrt{v} Z.$

\begin{definition}[Gaussian variance-mean mixture]\label{defvariancemeanmix}

The Gaussian variance-mean mixture distribution is the distribution of the increment of the Brownian motion process $\{Y_t, t \geq 0 \}$ (with drift $\theta$ and volatility $\sigma^2$ ) realized over a stochastic time increment $V$, a non-negative random variable independent of the Brownian motion process. The random variable 
\begin{equation}\label{variancemeanmix}
        X =  \theta V + \sigma \sqrt{V}Z
    \end{equation}
    is known as a Gaussian variance-mean mixture with mixing distribution $V$, denoted $GVM(\theta, \sigma, V).$
\end{definition}

The distribution of $X$ is fully characterized by the characteristic functions of $V$ and $Z$, which we explore in section \ref{SectionTSD}.  Proposition \ref{aVGcf} below details this for subordinated Brownian motion. The mean of the Gaussian variance-mean mixture  distribution is
$$\mathbb{E}[X] = \mathbb{E}_V [ \mathbb{E}_{X \mid V} [X]]] = \mathbb{E}_V[ \theta V] = \theta \mathbb{E}[V].$$
When $\theta=0$, the distribution is symmetric and is a {\it simple} variance mixture of Gaussian distributions with constant zero mean.  When $\theta \not = 0$, the symmetry of the distribution will depend on the symmetry of the mixing distribution of $V$. The mean and variance of $X$  each depend  on the distribution $V$ with their scales being linear functions of $V.$ When modelling daily log prices of stocks with annual units of time, the $V$ distribution corresponds to very small time increments $\approx 1/365$  or $1/252$. For such cases, the mean $ \theta V$ will be second order to the $\sqrt{V}Z$ (hence relatively negligible) and the distribution will be nearly symmetric. 
 
Gaussian variance-mean mixtures can be applied in empirical modelling of the price dynamics of financial assets or indices. Consider a \textit{discrete} time series of asset prices:
$\{p_0, p_1, p_2, \dots, p_n\} \,$
where $ \{t_0, t_1, \dots, t_n \}$ are the time points of the respective prices. With low frequency data, the time points are calendar dates and the prices are end-of-day prices.  If $y_j = \log(p_j)$, $j=0, \dots, n$ represents the time series of  log prices then 
$$ y_j = y_0 + \sum_{i=0}^j x_j,$$ 
where $x_j$ is the log price return
over the time interval $(t_{j-1}, t_j]$: $x_j = \log(p_j/p_{j-1}), \ j=1, \dots, n.$

It is common to let $j$ index the market days, weeks, or months depending on the frequency of the time series. Empirical modelling of the price dynamics of an asset treats the observed log price returns as realizations of random variables
$\{X_j\}_{j=0}^n.$ The simplest models assume that the $X_j$ log price returns are independent and identically distributed. A special case of this definition  introduced by Barndorff-Nielsen et al. \cite{Variancemeanmixtures} assumes further that $V$ is an infinitely divisible random variable. With this assumption, the discrete time stochastic process model can be embedded in a continuous-time model which has consistent specifications for alternate observation frequencies.

We would like to use the Gaussian variance-mean mixtures to be able to define time-subordinated Brownian motion processes 
\begin{equation}\label{TSBMearly}
    X_t = \theta \tau_t + \sigma W_{\tau_t},
\end{equation} where $\tau_t$ is a subordinator process independent of the Wiener process $W_t$. Our mixing distribution $V$ above will later characterise the stationary increments for the subordinating process $\tau_t$.
In this manner, the asymmetry term $\theta V$ relates to the random drift term $\theta\tau_t$, and the variance-mixing term $\sigma \sqrt{V}$ relates to the random noise term $\sigma W_{\tau_t}$. This will be explored in depth in section \ref{SectionTSBM}.
\begin{remark}\label{SVGdefremark}
    First consider the case $\theta=0$. Here $X$ becomes a symmetric distribution. Many of the potential $V$ distributions one considers are often 2+ parameter distributions making the parameters for the distribution $X$ non-identifiable if no further restrictions are imposed. However, this actually provides flexibility in interpreting properties of the distribution across different settings by choosing different restrictions.
    \begin{itemize}
        \item  If we restrict $V$ to have expectation $\mu = 1$, then $\sigma$ becomes a volatility parameter. This can be seen by using the independence of $V$ and $Z$,
        \begin{equation}
            \text{Var }[X] = \mathbb{E}\left[\left(\sigma\sqrt{V}Z\right)^2\right] = \sigma^2\mathbb{E}\left[{V}Z^2\right] = \sigma^2\mathbb{E}[V]\mathbb{E}\left[Z^2\right] = \sigma^2.
        \end{equation}
        This is a particularly useful notion which we will employ in section \ref{SectionTSBM} to define time-subordinated Brownian motion processes. 
        
        \item Alternatively, we could fix $\sigma$ (e.g $\sigma=1$) and directly observe changes to $V$ across different settings e.g. over different time periods for the process $X_t$.
    \end{itemize}

    For the case $\theta\neq0$, $X$ is no longer a symmetric distribution due to the $\theta V$ term. Further, $\sigma$ no longer represents the volatility as $\text{Var}[X] = \theta^2\text{Var}[V] + \sigma^2\mathbb{E}[V]$. However, if we let $\mathbb{E}[V] = \mu = 1$ and denote $\text{Var}[V] = \nu$ then the volatility is given by $\theta^2\nu + \sigma^2$. So the inclusion of the stochastic drift term determines the expectation just as a deterministic drift would, but additionally provides a further contribution to the variance above $\sigma^2$ unlike a deterministic drift.
\end{remark}

\begin{examples}
    Specifying different distributions for $V$ gives rise to various different Gaussian variance-mean mixture variables.
    \begin{itemize}
        \item Variance-Gamma ($V$ is Gamma distributed): We will refer back to this example throughout 
        \item Normal Inverse Gaussian ($V$ is Inverse Gaussian distributed)
        \item Generalised Hyperbolic ($V$ is Generalised Inverse Gaussian distributed)
    \end{itemize}
\end{examples}

In general there are no closed forms for the densities of these distributions, though there are special cases e.g. the Laplace distribution as a special case Variance Gamma distribution (with exponentially distributed $V$ as a special case of the Gamma distribution). In the general case we can condition the density on the distribution of $V$: for a standard Gaussian density $\phi(x)$ and stochastic variance density $f_V(v)$, one writes
\begin{equation}\label{General VG Density}
    f_X(x) = \int_{}^{}\phi\left(\dfrac{x-\theta v}{\sigma\sqrt{v}}\right)f_V(v)dv.
\end{equation}
For example, in the case of the Variance gamma distribution where $V\sim \text{Gamma}(\alpha,\beta)$, this becomes
\begin{equation}
    f_X(x) = \int_{0}^{\infty}\frac{1}{\sqrt{2\pi\sigma^2v}}e^{-(x-\theta v)^2/{(2\sigma^2v)}}\dfrac{\beta^\alpha}{\Gamma(\alpha)}v^{\alpha-1}e^{-\beta v}dv.
\end{equation}
Alternatively, the characteristic function has a much nicer form and, as we will explore, has strong applications.
\begin{prop}\label{aVGcf}
    Let $X = \theta V + \sigma\sqrt{V}Z$ be a Gaussian variance-mean mixture as in Definition \ref{defvariancemeanmix}. Then the characteristic function of the $X$, $\psi_X$ is given by
    \begin{equation}
        \psi_{X}(t) = \mathbb{E}_V\left[e^{(it\theta- t^2\sigma^2/2)V}\right].
    \end{equation}
\end{prop}
\begin{remark}\label{conditioning}
    Before providing the proof, we reflect on the Gaussian variance-mixture characterisation of $X$. When encountering such a mixture, we can first condition on the random variance to reduce the setup to Normality - something we understand well - then deal with the random variance separately. We will return to this idea throughout, later motivating the definitions of the variance-mixing transform and time-subordinator transform.
\end{remark}
\begin{proof}\label{VGcf}
   By conditioning on the value of $V$, we find for $t\in\mathbb{R}$,
   \begin{align}
       \psi_X(t) &= \mathbb{E}\left[e^{itX}\right] = \mathbb{E}\left[e^{it(\theta V+\sigma\sqrt{V}Z)}\right]=\mathbb{E}_V\left[\mathbb{E}_Z\left[e^{it(\theta V+\sigma\sqrt{V}Z)}\hspace{0.1cm}\Big\vert V\right]\right] \\
       &= \mathbb{E}_V\left[e^{it\theta V}\psi_Z\left(t\sigma\sqrt{V}\right)\right] = \mathbb{E}_V\left[e^{it\theta V- t^2\sigma^2V/2}\right] = \mathbb{E}_V\left[e^{(it\theta- t^2\sigma^2/2)V}\right]
   \end{align}
   where the last line follows from recalling that the characteristic function of a standard Gaussian is given by $\psi_Z(t) = e^{-t^2/2}$.
\end{proof}
\begin{cor}
    In the case that $V\sim \text{Gamma}(\alpha, \beta)$, we can complete the calculation to find
    \begin{align}
        \mathbb{E}_V\left[e^{(it\theta- t^2\sigma^2/2)V}\right] &=\int_{0}^\infty e^{(it\theta- t^2\sigma^2/2)v}\dfrac{\beta^\alpha}{\Gamma(\alpha)}v^{\alpha-1}e^{-\beta v}dv = \int_{0}^\infty \dfrac{\beta^\alpha}{\Gamma(\alpha)}v^{\alpha-1}e^{-(\beta-it\theta+t^2\sigma^2/2)v} dv,\label{useV}
    \end{align}
    after which we apply the definition of the Gamma function (noting $\beta+t^2\sigma^2/2>0$),
    \begin{equation}
        \int_{0}^\infty\dfrac{\beta^\alpha}{\Gamma(\alpha)}v^{\alpha-1}e^{-(\beta-it\theta+t^2\sigma^2/2)v} dv = \dfrac{\beta^\alpha}{\Gamma(\alpha)}\dfrac{\Gamma(\alpha)}{\left(\beta-it\theta+t^2\sigma^2/2\right)^\alpha} = \left(1 - \frac{i\theta t}{\beta} + \dfrac{t^2\sigma^2}{2\beta}\right)^{-\alpha}.
    \end{equation}
\end{cor} 

\begin{remark}\label{chartochar}
    From Proposition \ref{aVGcf}, we can find a useful relation between the characteristic functions of $X$ and $V$ above. This provides potential to find a general reverse direction: in particular, from a distribution $X$, we may find a $V$ satisfying $X\overset{d}{=}\theta V + \sqrt{V}Z$ (fixing $\sigma=1$ in accordance with remark \ref{SVGdefremark}).

    If we wish to consider the class of Gaussian variance-mean mixtures, it is natural to ask how the properties of the mixing distribution affect the full distribution. 
\end{remark}
Suppose the distribution of $X$ is known. Consider defining the distribution of $V$ implicitly through $X {d \over =} \theta V + \sqrt{V} Z,$ with $V$ and $Z$ independent. Proposition \ref{aVGcf} yields a candidate for the characteristic function of $V$ which could specify distribution of $V$ by the Levi continuity theorem for characteristic functions. 

Define the characteristic function of $V$ by $\psi_V(\omega) = \mathbb{E}[e^{i \omega V}]$, $\omega \in  R.$  For fixed $\omega_0\in\mathbb{R}$ consider informally finding the argument $t_0\in\mathbb{C}$ of $\psi_X(t)$ for which
"$\psi_V(\omega_0) = \psi_X(t_0)$".  For this to hold we must have
$(i\omega_0) = (i t \theta - t^2/2)$, a quadratic equation for $t$ with solutions:
$$t_0 = i \theta  \pm \sqrt{ (-1) ( \theta^2 + 2 i\omega_0) } = i\left( \theta \pm \sqrt{ \theta^2 + 2 i \omega_0}\right).$$
From Proposition \ref{aVGcf}, we have that $\psi_X(t) = \mathbb{E}[e^{itX}] = \mathbb{E}[e^{(it \theta - t^2 /2) V} ]$. We can then write
\begin{align}
    \psi_V(\omega) &= \mathbb{E}[e^{i\omega V}] = \mathbb{E}\left[e^{(-\theta+\sqrt{\theta^2+2i\omega})X}\right], \hspace{0.5cm}\omega\in\mathbb{R} \label{VcharX}
\end{align}
Either solution works and we consider using the principal root for $\sqrt{\theta^2 + 2 i \omega_0}$.  This equality is intuitive but not immediate as the domain of $\psi_X(\cdot)$ as a characteristic function is $\mathbb{R}$, whereas $t_0$ is complex - it remains to justify the existence of this expectation in (\ref{VcharX}). A proof of the existence of this expectation is given in Proposition \ref{existenceprop} in Section \ref{SectionTSD}. 

If we want to calculate a density $f_V(v)$ for $V$ from $f_X(x)$, we can compute $\psi_V(t)$ through relation (\ref{VcharX}) and then apply the characteristic function inversion formula,
\begin{equation}
    f_Y(x) = \psi^{-1}\left\{\psi_Y(\omega)\right\}(x) = \frac{1}{2\pi}\lim_{R\to\infty}\int_{-R}^{R}e^{-\omega^2/(2R^2)}e^{-i\omega x} \psi_Y(\omega)d\omega,
\end{equation}
which can be simplified in the case that $\psi_Y$ is itself $\mathcal{L}^1$ integrable or $f_Y$ is sufficiently smooth (i.e. satisfies the Dini criterion, see Katznelson \cite{Katznelson} Thm 2.5 for details),  
\begin{equation}
    f_Y(x) = \psi^{-1}\left\{\psi_Y(\omega)\right\}(x) = \frac{1}{2\pi}\lim_{R\to\infty}\int_{-R}^{R}e^{-i\omega x} \psi_Y(\omega)d\omega.
\end{equation}
The method outlines the formulation of an appropriate transform which encompasses the composition mapping $f_X(x)\mapsto \psi_X(\omega)\mapsto \psi_V(\omega)\mapsto f_v(x)$ .
\begin{definition}[Variance-mixing transform]\label{variancemixingtransform}
    Let $X$ be a Gaussian variance-mixture such that $X = \theta V+ \sqrt{V}Z$ for $\theta\in\mathbb{R}$. We define the generalised variance-mixing transform $\mathcal{V}^\theta[X]$ as
    \begin{equation}\label{ShenoyTransformTheta}
    \mathcal{V}^\theta\left[X\right](\xi) :=
    \psi^{-1}\left\{\mathbb{E}\left[e^{\left(-\theta+\sqrt{\theta^2+2i\omega}\right)X}\right]\right\}(\xi)
    \end{equation}
    In particular, $\mathcal{V}^\theta\left[X\right]$ gives the density of the random variable $V$ in the mixture.
\end{definition}
\begin{remark}\label{remarkTransform1}
    Employing $\mathcal{V}^\theta[X]$ empirically enables one to infer the stochastic volatility distribution $V$ directly from the log-return distribution $X$: rather than simply testing different distributions for $V$ and fitting to the data, we can directly infer $V$ from the price observations $X$ via a semi-parametric approach which does not rely on the quadratic variation.
    
    To calculate an empirical density for $f_V(v)$ based on a sample of $X$, one can map the sample onto the (empirical) characteristic function of $V$ through equation (\ref{VcharX}) and then apply a DFT to invert this characteristic function back into real space (for a review of empirical characteristic functions, the reader is directed towards Yu \cite{ECFYU})
    
    For the empirical transform $\mathcal{V}^{\theta}[X]$, $\theta$ is a hyper-parameter - one can scan a range of $\theta$ values, infer the distribution of $V$ and optimise the value of $\theta$ using the reconstruction error as the optimisation criterion. Alternatively, one could explore suitable values of $\theta$ such as $\theta = 0$ (so that the transform infers $V$ through the symmetric case $X = \sqrt{V}Z$). 
    We carry out the calculation of $\mathcal{V}^0[X]$ for S$\&$P500 log-returns data (January 2022 to January 2024) showing the results in Figure \ref{ecffigure}. The resulting density function appears to be relatively consistent with a Gamma distribution, and justifies the consideration of the Gamma distribution as our example stochastic variance $V$.

    It is possible that a further study can be carried out by tuning $\theta$ appropriately - this would be particularly valuable for cases when the skewness is significantly different from $0$, and one cannot reasonably assume a symmetric log-return distribution.
\end{remark}
\begin{figure}[h!]
    \centering
    \subfloat[\centering]{{\includegraphics[width=0.45\linewidth]{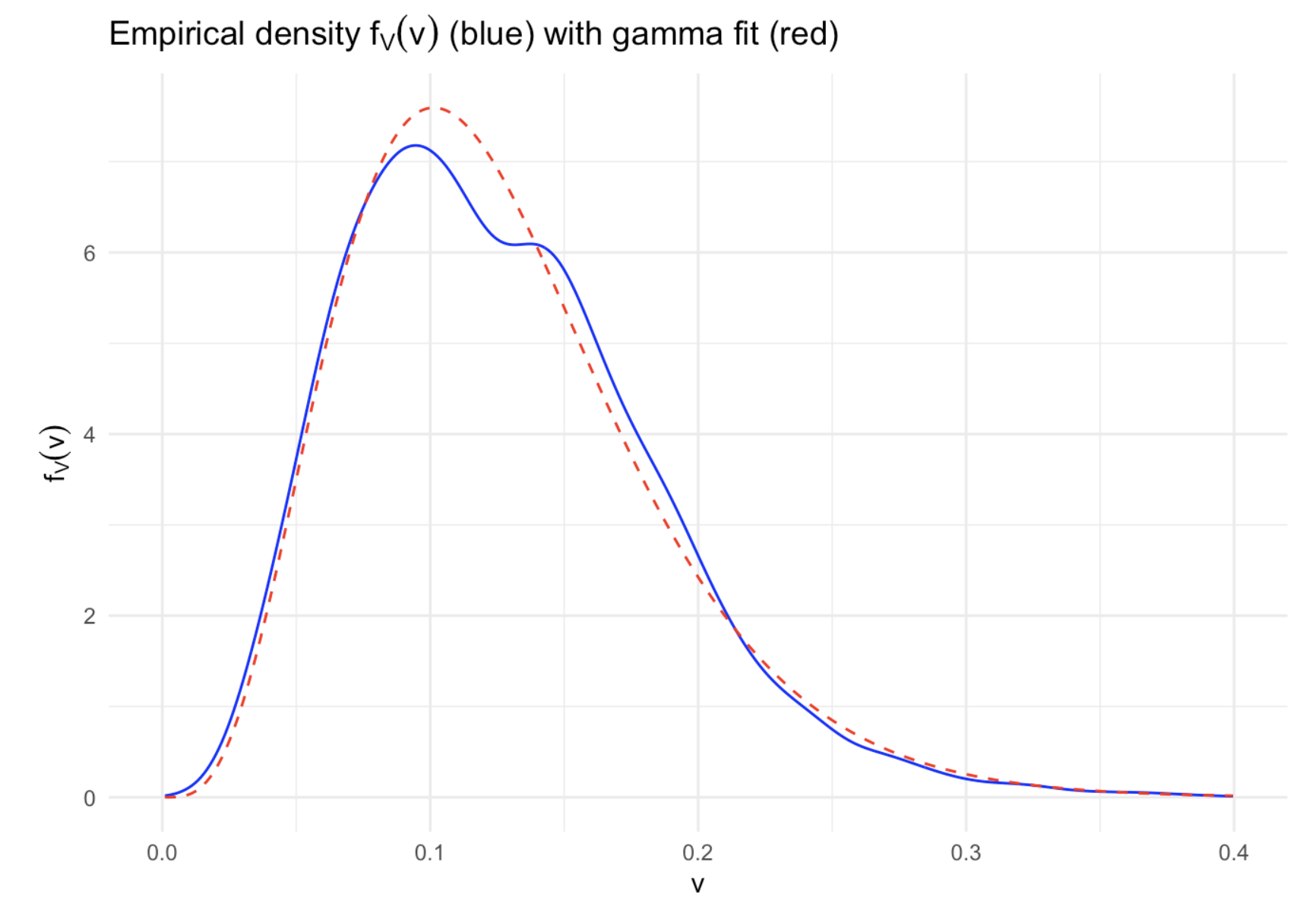}}}
    \qquad
    \subfloat[\centering]{{\includegraphics[width=0.45\linewidth]{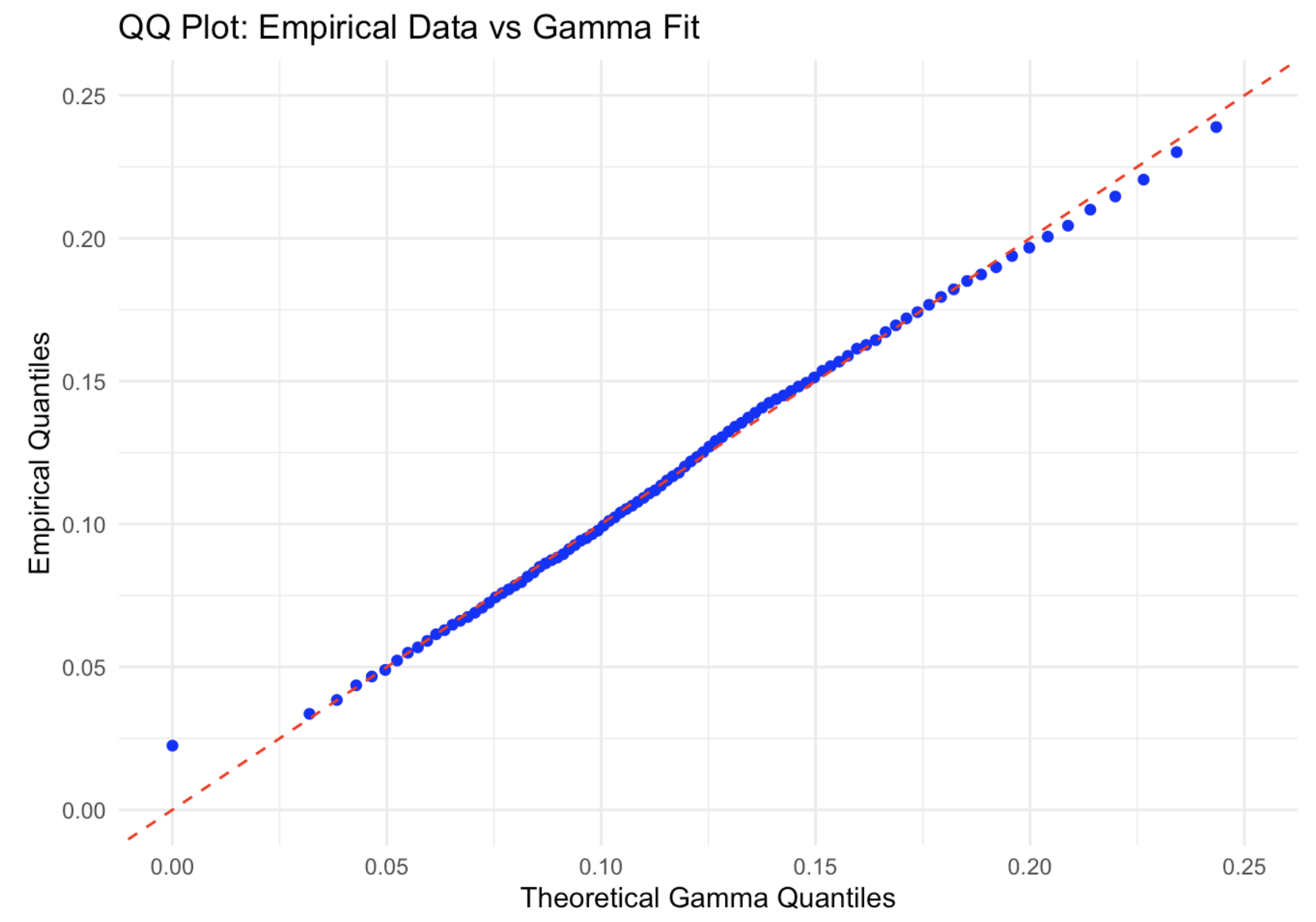} }}
    \caption{Gamma distribution fit for the stochastic variance $V = \mathcal{V}^0[X]$ for S$\&$P500 daily log returns data $X$ (January 2022 to January 2024)}
    \label{ecffigure}
\end{figure}

\section{Time-Subordinated Brownian Motion Processes}\label{SectionTSBM}

\label{SectionVGProcess}
We have characterised Gaussian variance-mean mixture distributions and now turn to the associated time-subordinated Brownian motion processes. This will follow a similar analysis, considering the Lévy process as a composition mixture of a Gaussian process (a Brownian motion) and a subordinator process as described in Section 1.3 of Applebaum \cite{Applebaum}. First, we clarify our definition of a Brownian motion, as given in 1.2 of Pitman and Yor \cite{pitmanyor}.
\begin{definition}[Brownian motion]\label{Brownian}
    Let $(W_t)_{t\geq0}$ be a standard Wiener process where $W_0 = 0$ a.s. , $W_t$ has stationary Gaussian increments $W_{t+h} - W_t \sim N(0, h)$ where the distributions are independent over non-overlapping time intervals, and $W_t$ is continuous $a.s.$ . Then the process $(b_t)_{t\geq0}$ given by
    \begin{equation}\label{Brownian motion}
        b_t = b(t;\theta, \sigma) = \theta t + \sigma W_t
    \end{equation}
    is a Brownian motion with drift $\theta\in\mathbb{R}$ and volatility $\sigma>0$.
\end{definition}
\noindent It is useful to recall the characteristic function for a Brownian motion,
\begin{equation}\label{bmcf}
    \psi_{b(t;\theta,\sigma)}(u) = \exp\left(i\theta tu -\frac{1}{2}\sigma^2t u^2 \right).
\end{equation}
We a characterised variance-mean mixture as a Gaussian with a stochastic mean and variance. In a Brownian motion, the increments have a mean and variance proportional to the size of the time increment, so to create an analogue we seek a way of evaluating the Brownian motion at stochastic, $V$ distributed time-increments - this is known as a \textit{time-subordinated} Brownian motion.
\begin{definition}[Subordinator process]\label{Gamprocess}
    Suppose $V$ is an infinitely divisible non-negative distribution. Let $(\tau_t)_{t\geq0}$ be a subordinator process with increments $V_h$. Then $\tau_0 = 0$ $a.s.$ , $\tau_t$ has stationary $V$ distributed increments $\tau_{t+h}-\tau_t\sim V_h$ where the distributions are independent over non-overlapping time intervals, and $\tau_t$ is continuous a.s. .
\end{definition}

\begin{example}
We can specify the Gamma process $(\gamma_t)_{t\geq0}$. It is useful to employ the parametrisation of the Gamma Process in terms of the mean and variance $\mu = \alpha/\beta, \nu = \alpha/\beta^2$ resp. so that we write $\text{Gamma}(\alpha,\beta)\sim\Gamma_{\mu,\nu}$. The Gamma process then has increments $\gamma_{t+h}-\gamma_t\sim \Gamma_{\mu h,\nu h}$ so that the cumulative value at time $t=T$ has distribution 
$\gamma_T \sim \Gamma_{\mu T, \nu T}$.
\end{example}

Now, we specify $V$ so that $\mathbb{E}[\tau_t] = t$. Recall from remark \ref{SVGdefremark} that this does not cause a loss of generality in characterising time-subordinated processes below. This is a helpful characterisation as it means that at time $T$, the expectation of our random Gamma distributed time $\mathbb{E}[\tau_T] = T$ and we have random fluctuation only about the deterministic time.


We can finally characterise time-subordinated Brownian processes as a composition of a Brownian motion and a subordinator process.
\begin{definition}[Time-subordinated Brownian motion]
    Let $(b_t)_{t\geq0} = b(t;\theta, \sigma)$ be a Brownian motion as in Definition \ref{Brownian},  and let $(\tau_t)_{t\geq0}$ be an independent subordinator process as in Definition \ref{Gamprocess}. Then the time-subordinated process $(X_t)_{t\geq0}$ given by the composition
    \begin{align}
        X_t = b_{\tau_t}
        &= \theta \tau_t + \sigma W_{\tau_t},\label{thetagammasigmawienergamma}
    \end{align}
    is known as a time-subordinated Brownian motion process.
\end{definition}
\begin{remark}
    Breaking down this definition, there are two sources of randomness. First is the Wiener process which acts the same as in the Brownian motion, creating Gaussian noise about the drift $\theta t$. However, there is also the random time-subordination, so that we think of the process as first randomly mapping deterministic time onto a subordinated time and then evaluating the Brownian motion at the time change. 
    Note that one could equivalently understand the process as observing a Brownian motion at randomly distributed times, those times given by $\tau_t$. 
    For a more detailed review of time-subordinated processes, see section 1.3 of Applebaum \cite{Applebaum}.
    
    Figure \ref{VGcomponents} illustrates a simulation
    \footnote{ 
        $10^6$ points used to simulate $b(t)$ and $\tau(t)$ over 10 units of time i.e. a discretization grid $\mathbb{T}$ with $\Delta t = 10^{-5}$. 
        
        $b(\tau(t))$ is evaluated by interpolation: first evaluate $\tau(t)$ for each point $t\in \mathbb{T}$, $\tau(t)$ then lies between $t_-, t_+\in \mathbb{T}$ and $b(\tau(t))$ is given via interpolation between $b(t_-)$ and $b(t_+)$.
    } 
    of a time-subordinated Brownian motion - specifically a Variance Gamma process. The original realisation of the Brownian motion $b(t)$ is given in the top left, and the realisation of an independent Gamma process $\tau(t)$ is given in the top right. The Subordinated Brownian motion in the two lower sub-figures is given by evaluating the original motion with the time change i.e. $(t,b(t)) \mapsto (t, b(\tau(t)))$.

    \begin{itemize}
        \item The shading on the left plots demonstrate the effect of the time-subordination. Each band in the Brownian motion plot is mapped onto the corresponding band in the Subordinated motion plot through the time change (i.e. the endpoints of the bands in the top plot are mapped onto the endpoints of the bands on the bottom plot via $t\mapsto \tau(t)$).
        \item While the general shape of the motion is mostly preserved, one observes the dilation and contraction of the Brownian motion along the time-axis caused by the time-subordination.
        \item The colour on the right plots demonstrate the `speed' of the time change over different increments. This is calculated as the mean rate of the Gamma process realised over the increment (i.e. the ratio of the gamma increment to the length of the increment). `Speed' 1 indicates there is no time change on average, `speed' 2 indicates that the time change moves twice as fast as normal time, `speed' 0.5 indicates that the time change is half as fast.
        \item We see that periods of extreme volatility occur where there is a contraction in the time axis and a higher `speed' of time change (in red). Alternatively, periods of lower volatility correspond to instances of time dilation and a slower `speed' of time change (in blue)
    \end{itemize}

    \begin{figure}[h!]
        \centering
        \includegraphics[width = \linewidth]{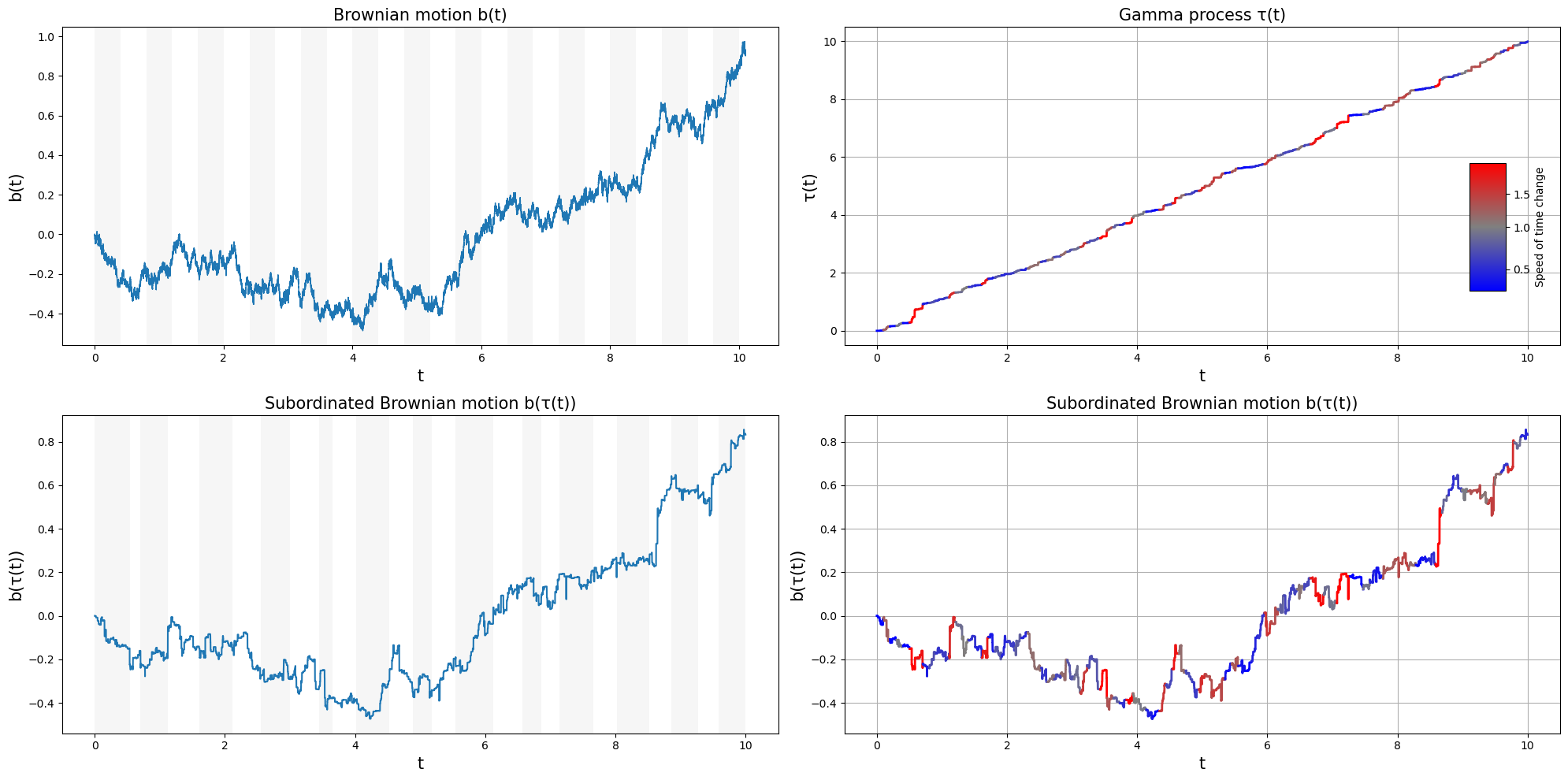}
        \caption{Realised components of a Time-subordinated Brownian motion process: $VG(t;\theta = 0.1, \sigma = 0.05, \nu = 0.05)$}
        \label{VGcomponents}
    \end{figure}
\end{remark}

We can obtain the density function of the process at time $t$ by conditioning on a realisation of the time-subordinator (giving  a Brownian motion density), and integrating against the subordinator density. Similar to (\ref{General VG Density}), we have
\begin{equation}
     f_{X_t}(x) = \int\phi\left(\dfrac{x-\theta \xi}{\sigma\sqrt{\xi}}\right)f_{\tau_t}(\xi)d\xi.
\end{equation}
We may also obtain the characteristic function for the process. Again, we condition on realisation of the time-subordinator with a similar method to the proof of Prop. (\ref{VGcf}), using the Brownian motion characteristic function (\ref{bmcf}). We then find
\begin{equation}\label{VGc.f.}
    \psi_{X_t}(u) = \mathbb{E}_{\tau_t}\left[e^{(iu\theta- u^2\sigma^2/2)\tau_t}\right].
\end{equation}

\begin{example}

Specifying the Gamma subordinator process we have the pair,
\begin{align}
    f_{X_t}(x) &= \int_{0}^\infty\dfrac{1}{\sqrt{2\pi\sigma^2g}}\exp\left(-\frac{(x-\theta g)^2}{2\sigma^2g}\right)\dfrac{g^{\frac{t}{\nu}-1}\exp\left(-\frac{g}{\nu}\right)}{\nu^{\frac{t}{\nu}}\Gamma\left(\frac{t}{\nu}\right)}dg\\
    \psi_{X_t}(u) &= \left({1-i\theta\nu u + \frac{\sigma^2}{2}\nu u^2}\right)^{-\frac{t}{\nu}}.
\end{align}

Taking the limit as $\nu\downarrow0$ in the $VG$ process, we obtain weak convergence to a Brownian motion $b(t;\theta,\sigma)$ as in (\ref{Brownian motion}) - one concludes this by observing the characteristic function above converges to that of a Brownian motion and applying Lévy's continuity Theorem. This should be expected, as the Gamma process subordinating the time becomes degenerate with $0$ variance about the mean $t$, and there is no time-change for the Brownian motion. This shows that the $VG$ model still contains Brownian motion as a sub-model - this is a particularly useful property if one employs a hypothesis test which nests a Brownian motion model null inside a more general $VG$ model alternative, observing whether $\nu$ is significantly greater than $0$. This result can similarly be generalised for other subordinator processes parametrised in terms of their variance $\nu$, and taking the limit $\nu\to 0$ appropriately.

One could compute the central moments of the $VG$ process from the characteristic function (\ref{VGc.f.}); however, employing the form (\ref{thetagammasigmawienergamma}) simplifies calculations. Conditioning on the Gamma time-subordinator as $\gamma_t = g$, the conditional $VG$ process is a Brownian motion and we write
\begin{equation}
    X\left(t;\theta, \sigma \big\vert \gamma_t=g\right) = \theta g +\sigma W_g.
\end{equation}
We can then take the expectation over the Wiener process, and then the expectation over $g=\gamma_t$ as in A4 of Madan et al. \cite{MadanCarrChang}). 
We obtain,
\begin{align*}
    &\mathbb{E}\left[X_t\right] = \theta t, \hspace{0.2cm}\\&\mathbb{E}\left[\left(X_t-\mathbb{E}[X_t]\right)^2\right] = (\theta^2\nu+\sigma^2)t, \hspace{0.2cm}\\&\mathbb{E}\left[\left(X_t-\mathbb{E}[X_t]\right)^3\right] = (2\theta^3\nu^2+3\sigma^2\theta\nu)t,\\ &\mathbb{E}\left[\left(X_t-\mathbb{E}[X_t]\right)^4\right] = (3\sigma^4\nu+12\sigma^2\theta^2\nu^2 + 6\theta^4\nu^3)t + (3\sigma^4+6\sigma^2\theta^2\nu+3\theta^4\nu^2)t^2.
\end{align*}
If we consider the case $\theta = 0$, then we have no skewness as the third central moment becomes $0$. Furthermore, the kurtosis would then be given by $3(1+\nu/t)$ so that $\nu$ represents the percentage excess kurtosis over the Brownian motion process over a unit time increment (where over longer time increments the excess kurtosis decreases to $0$ linearly in $t$ - an expected result considering the central limit theorem). 

However, for $\theta\neq0$ the variance of the process is higher than that of the Brownian motion - this increase in average volatility stems from additional variance contributions from the random drift component under the time-subordination (being deterministic, the drift component in a Brownian motion does not contribute to the variance). This is key to differences in option prices under a regular GBM model and a time-subordinated model - in particular, long volatility options are often priced higher under a time-subordinated model.

\end{example}

\section{Time-Change Decomposition}\label{SectionTSD}
\label{sectionTST}
In a similar manner to the variance-mixing transform through which one can transform a Gaussian mean-variance mixture $X = \theta V + \sqrt{V}Z$ into its stochastic variance $V$, one can transform a subordinated Brownian motion directly into its subordinating time-change process. This follows the same procedure as the variance-mixture transform in Definition \ref{variancemixingtransform}. Here, we observe the evolution of the distribution for the subordinating time-change process over time-increments of varying size.

Suppose we have a time-changed Brownian motion $(X_t)_{t\geq0}$ such that $\forall t\geq0$, $X_t = \theta \tau_t + W_{\tau_t}$ for some $\theta\in\mathbb{R}$ where $(W_t)_{t\geq0}$ is a standard Wiener process and $(\tau_t)_{t\geq0}$ is the independent time-change. If we want to compute the density $f_{\tau}(\xi,t)$ of the time-change, we can use the \textit{time-change transform} defined below. This is motivated by a combination of Definition \ref{variancemixingtransform} and the theory of Laplace exponents for Lévy processes, which take a similar form, characterising a Lévy process in terms of its Laplace transform (see 1.3.2 of Applebaum \cite{Applebaum}).

It is particularly important to highlight case that $\theta = 0$, i.e. the no-drift case where $X_t = W_{\tau_t}$ through the following theorem of Skorokhod \cite{skororussia}, translated in \cite{skorokhod}, which asserts that any probability distribution can be represented as a stopped Brownian motion.

\begin{thm}\label{Skorothm}
    For a given probability measure $\mu$ on $\mathbb{R}$ such that $$\int_\mathbb{R} |x|d\mu(x) < \infty, \hspace{0.5cm} \int_\mathbb{R}xd\mu(x) = 0,$$
    there exists a stopping time $T$ such that $W_T \sim \mu$ and the stopped process $(B_{t\wedge T})_{t\geq0}$ is a uniformly integrable martingale.
\end{thm}
There is a vast literature of different solutions constructing the stopping time $T$ in the Theorem, as summarised by Obłój \cite{Obloj}. Extending the result, Monroe \cite{Monroe} showed the following key Theorem.
\begin{thm}\label{Monroethm}
    A process $(X_s, \mathcal{F}_s)$ can be embedded in Brownian motion if and only if $(X_s, \mathcal{F}_s)$ is a local semi-martingale. To embed a process in Brownian motion is to find a Wiener process $(W_t, \mathcal{G}_s)$ and an increasing family of $\mathcal{G}_s$ stopping times $T_s$, such that $W_{T_s}$ has the same joint distributions as $X_s$.
\end{thm}
This means that any semi-martingale $X_t$ can be represented as $W_{\tau_t}$ for some increasing process $\tau_t$ \footnote{It is however important to note that the stopping times $T_s$ do \textit{not} necessarily have to correspond to a Lévy subordinator process}. For our study, as explained by Veraart and Winkel \cite{VeraartWinkel}: `In the light of the Fundamental Theorem of Asset Pricing, this means that every arbitrage-free model can be viewed as time-changed Brownian motion' which provides a strong justification for considering this inverse problem for the cases we model the time-change with a subordinator. The transform below represents an alternate characterisation of the subordinator without computing the quadratic variation.

\begin{definition}\label{timesubordinator transform}
    Let $X = (X_t)_{t\geq0}$ be a time-changed Brownian motion where $X_t = \theta \tau_t + W_{\tau_t}$ for $\theta\in\mathbb{R}$. We define the time-change transform $\mathcal{S}^\theta[X]$ as
    \begin{equation}\label{ShenoySubordinatorTransformTheta}
    \mathcal{S}^\theta\left[X\right](\xi,t) :=
    \psi^{-1}\left\{\mathbb{E}\left[e^{\left(-\theta+\sqrt{\theta^2+2i\omega}\right)X_t}\right]\right\}(\xi).
    \end{equation}
    In particular, $\mathcal{S}^\theta\left[X\right]$ gives the density for the subordinator $\tau_t$.
\end{definition}

\begin{remark}
    Fixing the value of $t$, this definition is identical to that of the variance-mixture transform (Definition \ref{variancemixingtransform}): The inner expectation represents relation (\ref{VcharX}) mapping $X_t$ onto the characteristic function of $\tau_t$, while $\psi^{-1}$ represents the characteristic function inversion. The transform encompasses the composition mapping $f_{X}(x,t)\mapsto \psi_{X_t}(u)\mapsto \psi_{\tau_t}(u)\mapsto f_{\tau_t}(\tau)$ as shown in Figure \ref{Compositionmappingfigure} below.
    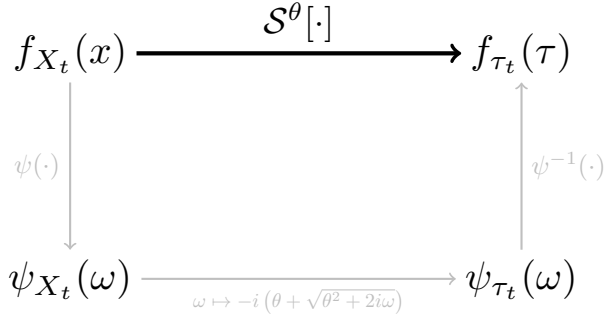
\begin{figure}[h!]
        \centering
        \begin{tikzpicture}[scale=1.5] 
            \node (A) at (0,2) {$\scaleto{f_{X_t}(x)}{15pt}$};
            \node (B) at (0,0) {$\scaleto{\psi_{X_t}(\omega)}{15pt}$};
            \node (C) at (4,0) {$\scaleto{\psi_{\tau_t}(\omega)}{15pt}$};
            \node (D) at (4,2) {$\scaleto{f_{\tau_t}(\tau)}{15pt}$};
        
            \draw[->, thick, draw = lightgray] (A) -- (B) node[midway, left] {$\color{lightgray}\psi(\cdot)$};
            \draw[->, thick, draw = lightgray] (B) -- (C) node[midway, below] {\scalebox{0.75}{$\color{lightgray} \omega\mapsto -i\left(\theta + \sqrt{\theta^2+2i\omega}\right)$}};
            \draw[->, thick, draw = lightgray] (C) -- (D) node[midway, right] {$\color{lightgray} \psi^{-1}(\cdot)$};
            \draw[<-, ultra thick, , draw = black] (D) -- (A) node[midway, above] {$\scaleto{\mathcal{S}^{\theta}[\cdot]}{15pt}$};

        \end{tikzpicture}
        \caption{Components of $\mathcal{S}^{\theta}[\cdot]$}
        \label{Compositionmappingfigure}
\end{figure}
    As discussed above, this allows us to transform the full process $(X_t)_{t\geq0}$ directly into its subordinator process $(\tau_t)_{t\geq0}$ via a semi-parametric approach. This transform allows one to explore properties of the hidden subordinating process $(\gamma_t)_{t\geq0}$ by observing only the price evolution process $(X_t)_{t\geq0}$. In effect, one reduces the study of the composition of two random objects (the Brownian motion and the subordinator) into the study of just one (the subordinator).
\end{remark}

The reader may question the existence of the transform $\mathcal{S}^\theta\left[X\right]$ as this is not an obvious result. The central issue lies in verifying that the expectation (\ref{verifythisexpectation}) below exists - this will be resolved by employing a conditional expectation and extending the notion of the characteristic function to have a complex domain as in Lukacs \cite{Lukacs}.

\begin{prop}\label{existenceprop}
    Suppose $X = \theta V + \sqrt{V}Z$ is a Gaussian variance-mean mixture. Then the transform $\mathcal{V}^\theta[X]$ given by
    \begin{equation}\label{variancemixtureverification}
    \mathcal{V}^\theta\left[X\right](\xi) :=
    \psi^{-1}\left\{\mathbb{E}\left[e^{\left(-\theta+\sqrt{\theta^2+2i\omega}\right)X}\right]\right\}(\xi)
    \end{equation}
    exists and equals $f_V(\xi)$, the density of $V$.
\end{prop}

We first want to verify the existence of the inner expectation
\begin{equation}\label{verifythisexpectation}
    \mathbb{E}\left[e^{\left(-\theta+\sqrt{\theta^2+2i\omega}\right)X}\right] = \mathbb{E}\left[e^{\left(-\theta+\sqrt{\theta^2+2i\omega}\right)\left(\theta V+\sqrt{V}Z\right)}\right]
\end{equation}
for all $\omega\in\mathbb{R}$, and show that it is equal to $\psi_V(\omega)$ which must exist as the characteristic function of a real random variable $V$. We can then conclude that the transform exists with the characteristic function inverse.

This expectation (\ref{verifythisexpectation}) has a similar form to the characteristic and moment generating functions; however, these cannot be applied directly as their domains only cover the imaginary and real axes respectively, while the necessary domain of integration in (\ref{verifythisexpectation}) is $\mathcal{I} = \{-\theta + \sqrt{\theta^2 + 2i\omega}\hspace{0.1cm}\big\vert \hspace{0.1cm}\omega\in\mathbb{R}\}$. This leads us to consider an \textit{extended} characteristic function which extends the domain of the characteristic function into the complex plane. For a more detailed review of analytic characteristic functions and their applications, the reader is directed to Chapter 7 of Lukacs \cite{Lukacs}.

\begin{definition}
    Let $Y$ be a real-valued random variable with density $f_Y(y)$, and $\mathcal{D}\supset\mathbb{R}$ be a region in the complex plane containing the real line such that the Fourier integral
    \begin{equation}\label{FourierIntegral}
        \Psi_Y(z) = \int_{\mathbb{R}}e^{iyz}f_Y(y)dy
    \end{equation}
    is analytic in $\mathcal{D}$. Then $\Psi_Y$ is the analytic continuation of $\psi_Y$ in the domain $\mathcal{D}$.
\end{definition}
In general, it is not possible to analytically extend the characteristic function to the entire complex plane. Indeed, any distribution with a moment-generating function which is undefined for some parts of the real line (such as the Variance-Gamma) cannot be extended to the entire complex plane, as there will be a point on the imaginary axis where the Fourier integral is undefined \footnote{As a simple example, the standard Laplace distribution (as a special case of the Variance Gamma) has density $f_L(x) = e^{-|x|}/2$ and characteristic function $\psi_L(z) = (1+z^2)^{-1}$ which has singularities at $\pm i$, and so there is no analytic continuation to the entire plane.}. To this end, Lukacs shows that analytic characteristic functions are regular in a horizontal strip (the strip of regularity), where there are singularities on the intersection of the boundary and the imaginary axis (though in some cases this strip is the entire plane, and there are no singularities). The characteristic function is represented by the integral (\ref{FourierIntegral}) inside this domain. Lukacs remarks that analytic characteristic functions can often be continued analytically beyond this strip of regularity if the singularities are suitably avoided, as we will require for our transform - we show that the Fourier integral \textit{is} is well-defined on $\mathcal{I} = \{-\theta + \sqrt{\theta^2 + 2i\omega}\hspace{0.1cm}\big\vert \hspace{0.1cm}\omega\in\mathbb{R}\}$ as well as $\mathbb{R}$.

Lukacs also shows that for distributions with moment-generating functions defined on the entire real line, the characteristic function can be extended to the entire complex plane. In particular, the Gaussian has this property -  we show explicitly the existence of an analytic continuation of the Gaussian characteristic function in the lemma below. This provides a particularly useful step towards our goal.

\begin{lemma}\label{ExtendedGaussianCf}
    Let $Z$ be a standard Gaussian. Then $\Psi_Z$ is entire in the complex plane $\mathbb{C}$ with $\Psi_Z(z) = e^{-z^2/2}$.
\end{lemma}
\begin{proof}
    Fix $z\in\mathbb{C}$. We begin by completing the square in the Fourier integral (\ref{FourierIntegral}),
    \begin{equation}\label{Fourierintegralstart}
        \int_{\mathbb{R}}e^{ixz}\dfrac{e^{-x^2/2}}{\sqrt{2\pi}}dx = e^{-z^2/2}\int_{\mathbb{R}} \dfrac{1}{\sqrt{2\pi}}e^{-(x-iz)^2/2}dx.
    \end{equation}
    This already has the desired form - once we show that the integral on the RHS is equal to $1$ we may conclude. Let $L>|z|$ and consider the finite integral,
    \begin{equation}
        I_L = \int_{-L}^L \dfrac{1}{\sqrt{2\pi}}e^{-(x-iz)^2/2}dx = \int_{\mathcal{C}_1} \dfrac{1}{\sqrt{2\pi}}e^{-\xi^2/2}d\xi,
    \end{equation}
    where $\mathcal{C}_1$ is the straight line segment $-L-iz\to L-iz$.
    Letting $L\to\infty$, $I_{\infty}$ is the integral we want to evaluate (the integral RHS of (\ref{Fourierintegralstart})). Let $\mathcal{C}$ be the closed loop contour $\mathcal{C} = \mathcal{C}_1\cup \mathcal{C}_2\cup \mathcal{C}_3 \cup \mathcal{C}_4$ encompassing the four line segments (a parallelogram) $$-L-iz\overset{\mathcal{C}_1}{\longrightarrow} L-iz\overset{\mathcal{C}_2}{\longrightarrow} L \overset{\mathcal{C}_3}{\longrightarrow} -L \overset{\mathcal{C}_4}{\longrightarrow} -L-iz.$$ We bound the integrals over $\mathcal{C}_2$ and $\mathcal{C}_4$ using the M-L inequality,
    \begin{align}
        \left|\int_{\mathcal{C}_2}\dfrac{1}{\sqrt{2\pi}}e^{-\xi^2/2}d\xi\right|\leq \dfrac{|z|}{\sqrt{2\pi}}\sup_{\xi\in\mathcal{C}_2}\left|e^{-\xi^2/2}\right|\leq \dfrac{|z|}{\sqrt{2\pi}}e^{-(L-|z|)^2/2},
    \end{align}
    with the same bound for the integral over $\mathcal{C}_4$. Taking the limit as $L\to\infty$, the bound goes to $0$ so that the integrals over $\mathcal{C}_2$ and $\mathcal{C}_4$ vanish as $L\to\infty$.  As $\phi(\xi) = e^{-\xi^2/2}/\sqrt{2\pi}$ is entire, its integral over the closed loop $\mathcal{C}$ is $0$ for all $|L|>z$. This means that in the limit $L\to\infty$,
    \begin{equation}
        \underbrace{\int_{\mathcal{C}_1} \dfrac{1}{\sqrt{2\pi}}e^{-\xi^2/2}d\xi}_{=I_\infty} + 0 + \int_{\mathcal{C}_3}\dfrac{1}{\sqrt{2\pi}}e^{-\xi^2/2}d\xi + 0 = 0.
    \end{equation}
    But for $L\to\infty$ the integral over $\mathcal{C}_3$ is just the Gaussian integral (in the reverse direction),
    \begin{equation}
        \int_{\mathcal{C}_3}\dfrac{1}{\sqrt{2\pi}}e^{-\xi^2/2}d\xi =  \int_{\infty}^{-\infty}\dfrac{1}{\sqrt{2\pi}}e^{-\xi^2/2}d\xi = -1.
    \end{equation}
    So indeed $I_{\infty} = 1$ and we conclude.
\end{proof}

With this in hand, we can return to the proof of Proposition \ref{existenceprop}.
\begin{proof}
Fix $\omega\in\mathbb{R}$. First, we condition\footnote{Clearly this assumes the original expectation exists. We argue this a posteriori.} the expectation on $V$,
\begin{align}
    \mathbb{E}\left[e^{\left(-\theta+\sqrt{\theta^2+2i\omega}\right)\left(\theta V+\sqrt{V}Z\right)}\right] &= \mathbb{E}_V\left[\mathbb{E}_Z\left[e^{\left(-\theta+\sqrt{\theta^2+2i\omega}\right)\left(\theta V+\sqrt{V}Z\right)}\big\vert V\right]\right]\\
    &= \mathbb{E}_V\left[e^{\left(-\theta^2+\theta\sqrt{\theta^2+2i\omega}\right)V}\mathbb{E}_Z\left[e^{ \left(-\theta+\sqrt{\theta^2+2i\omega}\right) \sqrt{V}Z}\big\vert V\right]\right]\label{conditionedsqrtvz}.
\end{align}
Using Lemma \ref{ExtendedGaussianCf}, the inner conditional expectation exists and is given by,
\begin{align}
    \mathbb{E}_Z\left[e^{ \left(-\theta+\sqrt{\theta^2+2i\omega}\right) \sqrt{V}Z}\big\vert V\right] &= \Psi_Z\left(-i\left(-\theta+\sqrt{\theta^2+2i\omega}\right)\sqrt{V}\right)= e^{-(-i)^2\left(-\theta+\sqrt{\theta^2+2i\omega}\right)^2V/2}\\
    &= e^{\left(\theta^2 + \theta^2+2i\omega-2\theta\sqrt{\theta^2+2i\omega}\right)V/2} = e^{i\omega V} e^{-\left(-\theta^2+\theta\sqrt{\theta^2+2i\omega}\right)V},
\end{align}
so that (\ref{conditionedsqrtvz}) reduces down perfectly to
\begin{equation}
    \mathbb{E}_V\left[e^{\left(-\theta^2+\theta\sqrt{\theta^2+2i\omega}\right)V}\mathbb{E}_Z\left[e^{\left(-\theta+\sqrt{\theta^2+2i\omega}\right)\sqrt{V}Z}\big\vert V\right]\right] = \mathbb{E}\left[e^{i\omega V}\right].
\end{equation}
As $V$ is a real random variable this expectation (the characteristic function of $V$) must exist so we have indeed verified that (\ref{verifythisexpectation}) exists for all $\omega\in\mathbb{R}$ as desired. We conclude the existence of the transform with the existence of the characteristic function inverse
\begin{equation}
    f_V(x) = \psi^{-1}\left\{E\left[e^{i\omega V}\right]\right\}(x) = \frac{1}{2\pi}\lim_{R\to\infty}\int_{-R}^{R}e^{-\omega^2/(2R^2)}e^{-i\omega x} \psi_Y(\omega)d\omega.
\end{equation}
\end{proof}
As discussed after Remark \ref{chartochar}, the inversion can be simplified in the case that $f_Y$ is smooth or satisfies the (weaker) Dini criterion at $x$:
\begin{equation}
    \int_{-1}^{1}\left\vert\dfrac{f_Y(x+t) - f_Y(x)}{t}\right\vert dt < \infty,
\end{equation}
(see Katznelson \cite{Katznelson} Theorem 2.5 for a justification), or $\psi_Y$ is itself $\mathcal{L}^1$ integrable. If any of these criteria are met, then we can recover the density with the simplified inversion formula
\begin{equation}
    f_Y(x) = \psi^{-1}\left\{\psi_Y(\omega)\right\}(x) = \frac{1}{2\pi}\lim_{R\to\infty}\int_{-R}^{R}e^{-i\omega x} \psi_Y(\omega)d\omega.
\end{equation}

\section{Example Empirical Analysis of a Decomposed Time-Change Process}\label{SectionEA}
Empirically, we observe the full distribution $X$ (the log returns data) and we wish to understand the subordinating, time-change process. We can fix a vector containing different values of time increments $t$ (holding periods), and then apply this transform as in remark (\ref{remarkTransform1}) to data across each time increment size to study the evolution of the subordinator process as the size of the time increment increases (e.g. for a log-return process, observing the evolution of the process over periods of 1 day, multiple days, or a week). This allows one to: a) characterise the subordinator process and b) observe whether the Lévy process assumptions are justified. Assuming that the value of $\theta$ is given (e.g. $\theta=0$ for a semi-martingale assumption in accordance with Theorem \ref{Monroethm}), this directly provides an explicit form for the evolution of the time-change. 

We show the results of naïvely applying the transform to daily S$\&$P500 log returns data by fixing $\theta=0$ in Figure \ref{Gammaprocessfits}. This is done by calculating the empirical characteristic function of $V$, $\hat{\psi}_V(t)$ from sample log returns $\{x_j\}_{j=0}^{N-1}$ and passing this through a DFT to return $\hat{f}_V(v)$.
\begin{figure}[h!]
    \centering
    \subfloat[\centering]{{\includegraphics[width=0.45\linewidth]{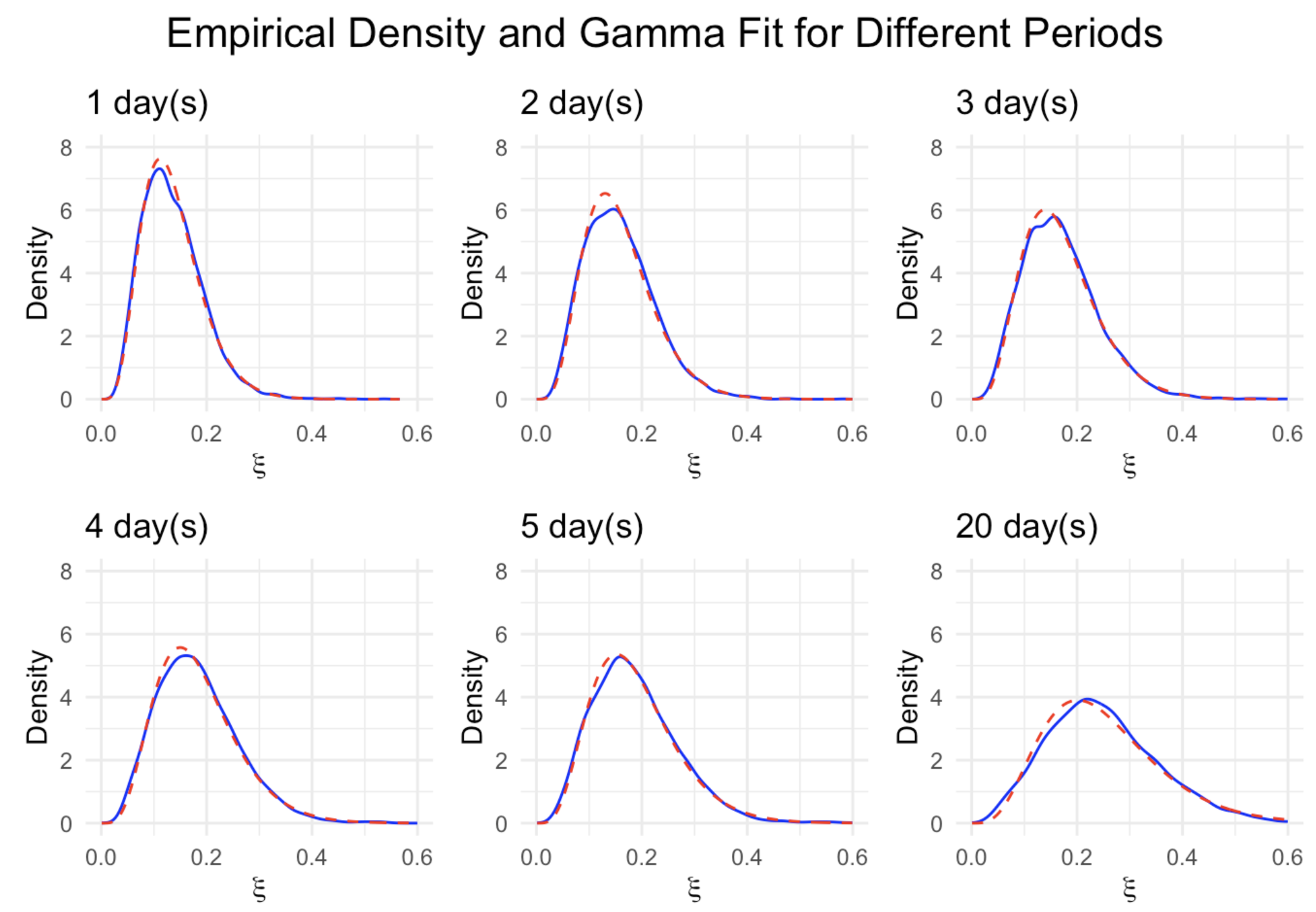} }}
    \qquad
    \subfloat[\centering]{{\includegraphics[width=0.45\linewidth]{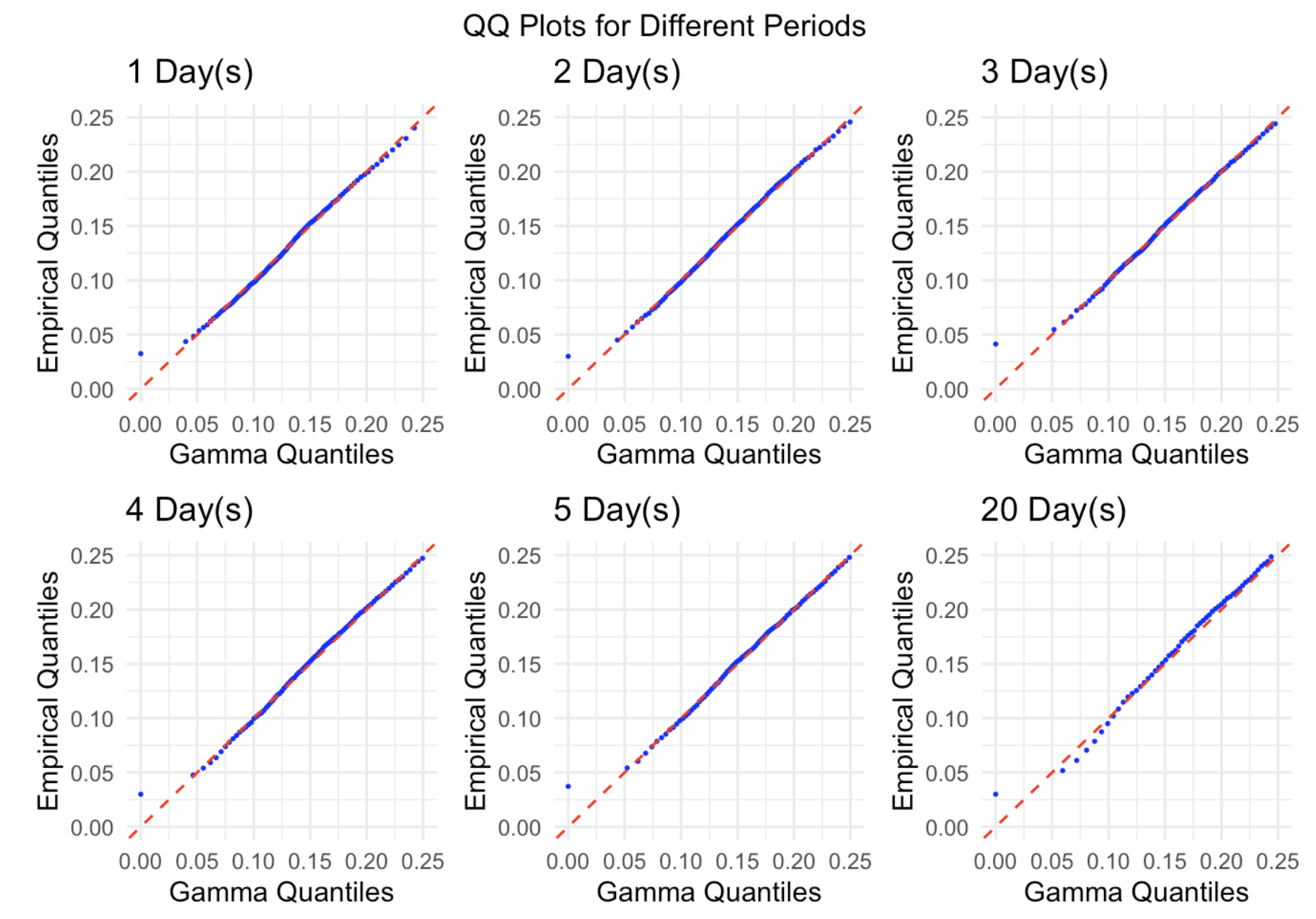} }}
    \caption{Fit of Gamma process as the subordinator process for S$\&$P500 daily log returns (January 2022 to January 2024) under the time-subordinator transform}
    \label{Gammaprocessfits}
\end{figure}
The resulting distributions appear relatively consistent with a Gamma distribution, and justify the consideration of the Gamma process as the subordinator process (though there is still room for improvement). Figure (a) also appears to show an increase in the mean and variance of the distributions as the size of the time-increments increases. We visualise the evolution of the empirical time-change process, and the fitted Gamma process over different sized time-increments in Figure \ref{Gammaprocessevolution}.
\begin{figure}[h!]
    \centering
    \subfloat[\centering]{{\includegraphics[width=0.46\linewidth]{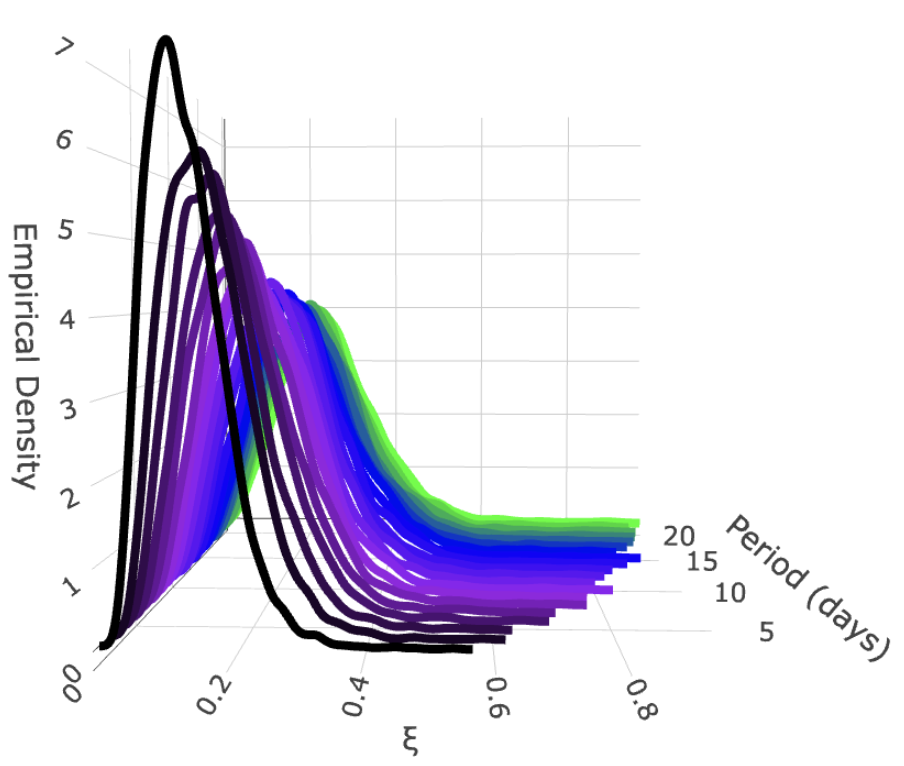} }}
    \qquad
    \subfloat[\centering]{{\includegraphics[width=0.46\linewidth]{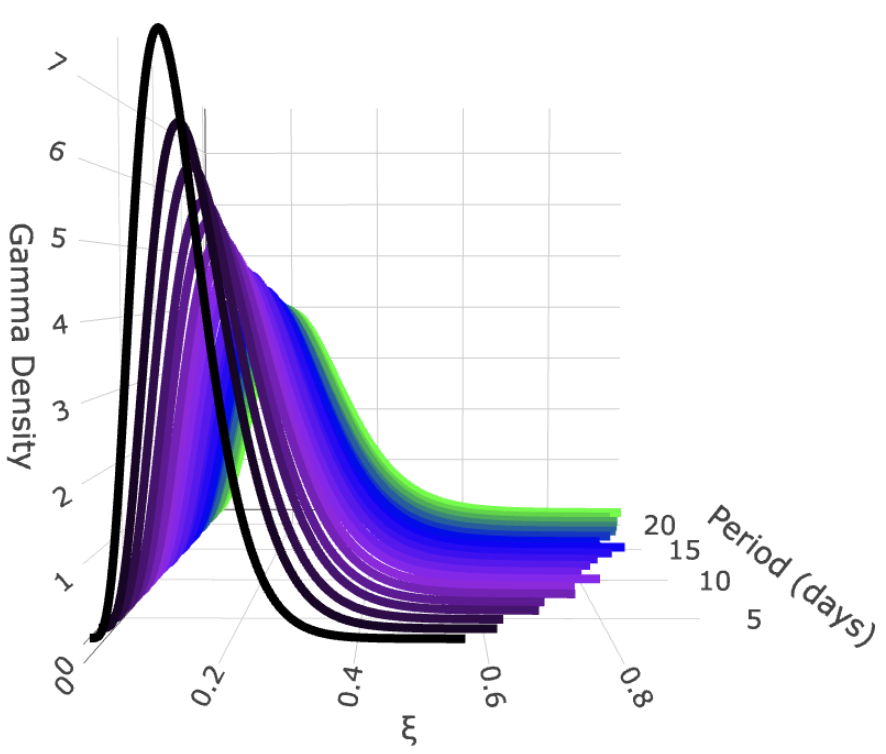} }}
    \caption{Evolution of (a) the empirical time-change process  and (b) the fitted Gamma process  for S$\&$P500 daily log returns (January 2022 to January 2024) from the time-change transform}
    \label{Gammaprocessevolution}
\end{figure}

Analysing this Figure, the mean does indeed increase with the size of the time-increments (days) of the process, and the variance around the mean increases similarly too. This is confirmed by evaluating $\mu$ and $\nu$ over the different time-periods in Figure \ref{Munu}.
\begin{figure}[h!]
    \centering
    \subfloat[\centering]{{\includegraphics[width=0.46\linewidth]{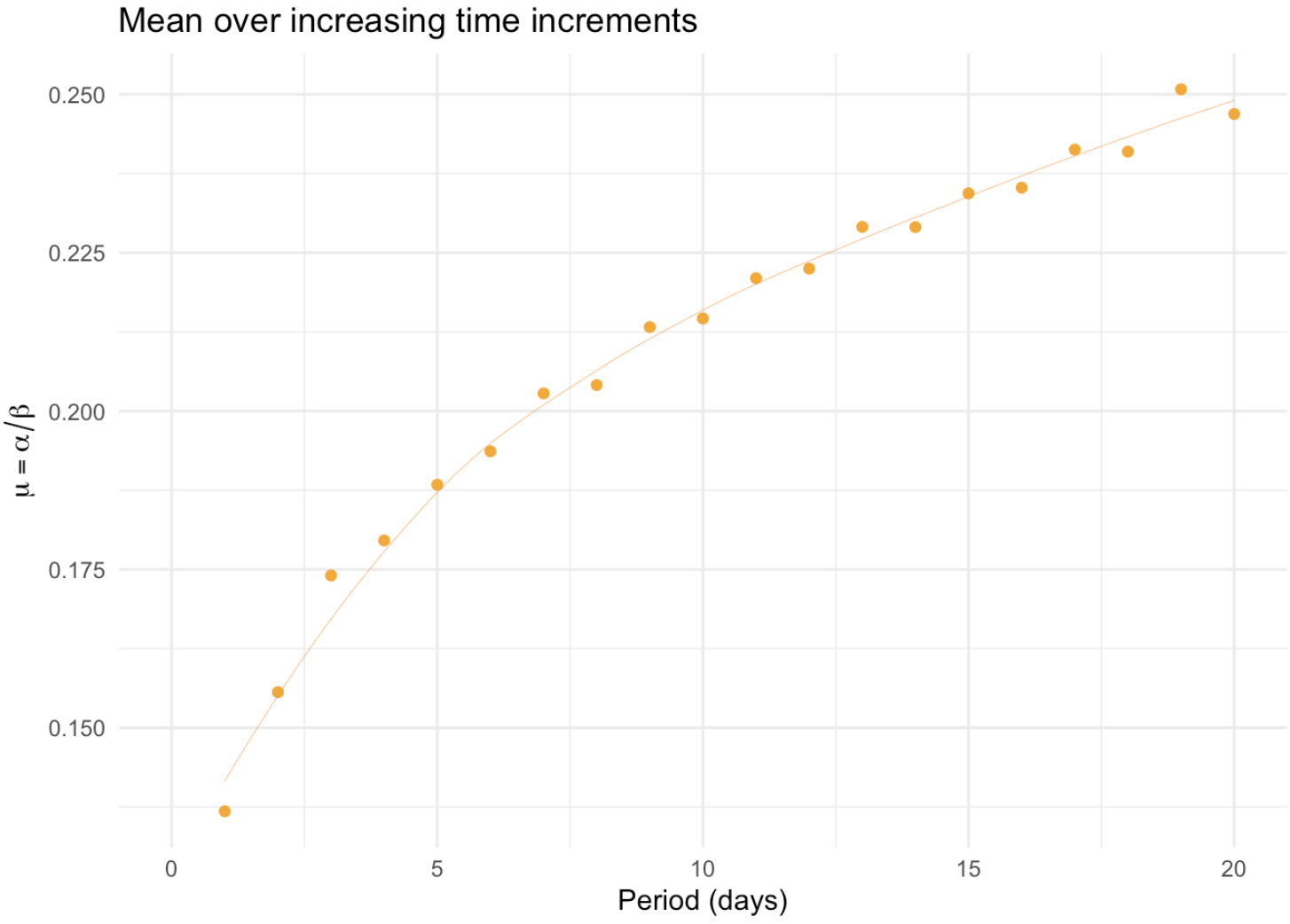} }}
    \qquad
    \subfloat[\centering]{{\includegraphics[width=0.46\linewidth]{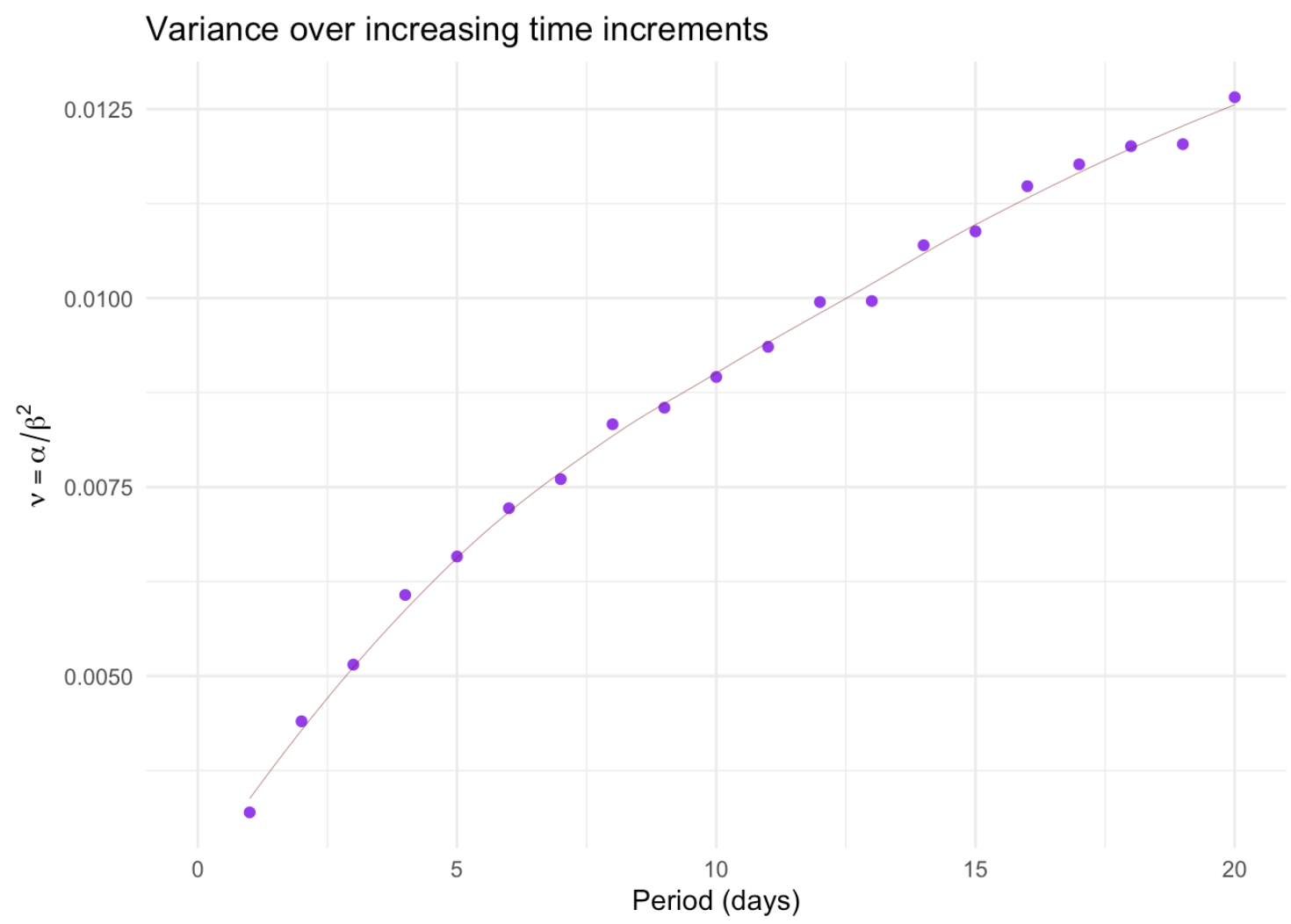} }}
    \caption{The mean and variance of the fitted Gamma process for $\theta = 0$}
    \label{Munu}
\end{figure}
However, the assumption that the mean and variance of the distribution of $\gamma_t$ increase linearly with time (i.e $\gamma_{t+h}-\gamma_t\sim \Gamma_{\mu h,\nu h}$) is not quite true - the rates of increase appear to slow over time. This suggests that simple Lévy subordinator assumptions may not be entirely appropriate as these would imply a linear trend in both variables over time, and that the true time-change dynamics (in correspondence with Theorem \ref{Monroethm}) are more nuanced.

These findings present good scope for future study. For example if we drop the semi-martingale assumption (and so cannot refer to Theorem \ref{Monroethm}) and don't assume the value of the hyper-parameter $\theta$, one again needs to tune this value as in the case for the variance-mixture transform (here the interested reader is directed towards Masuda (2014) \cite{masuda}). As the drift term interacts with both the mean and variance of the time-subordinator, perhaps we can expect improvement over longer time periods where the Lévy process assumptions may be more appropriate. One could even try to include time series dependence for these parameters as in Mercuri and Bellini \cite{MercuriBellini} (2010); however, for our study we wish to consider inference based on a simpler time-subordination mapping (i.e. modelling the economically relevant market `business time' as stochastic), and introducing time-series dependence complicates this viewpoint dramatically.

\section{Discussion}\label{SectionSummary}
This paper explores the concept of stochastic time change in modelling stock-price dynamics, first presenting results on Gaussian variance-mean mixtures and generalising to the class of continuous time-subordinated Brownian motion models. We explore the Variance Gamma process as a key example throughout.

We detail a Fourier method to consider the appropriate variance distribution in the Gaussian variance-mixture and formulate the time-subordinator transform which reduces the study of a subordinated Brownian motion process to the more focused study of its subordinating process by yielding an estimate for the subordinating distribution without using the quadratic variation. We ultimately prove the theoretical existence of this transform for application.

A brief empirical study of log-returns data from the S$\&$P500 market using this transform follows - we find the Gamma process appears to be an appropriate subordinator process for the Brownian motion for the data; however, the Lévy subordinator assumptions appear too restrictive, providing scope for future study into model extensions.

\section*{Acknowledgements}
We would like to acknowledge the Imperial-MIT International Research Opportunities Programme without which this research collaboration would not have been possible. Rohan Shenoy would like to acknowledge the financial support provided by the Mathematics Department at Imperial College London, the Imperial International Relations Office, and the UK Government's Turing Scheme.

\bibliographystyle{unsrt}
\bibliography{Bibliography.bib}


\end{document}